\documentclass{article}
\usepackage{ijcai25}
\usepackage{times}
\usepackage{soul}
\usepackage{url}
\usepackage[hidelinks]{hyperref}
\usepackage[utf8]{inputenc}
\usepackage[small]{caption}
\usepackage{graphicx}
\usepackage{amsmath}
\usepackage{amsthm}
\usepackage{amssymb}
\usepackage{booktabs}
\usepackage{algorithm}
\usepackage[noend]{algorithmic}
\usepackage[switch]{lineno}
\usepackage[skip=0pt]{caption}

\newtheorem{theorem}{Theorem}
\newtheorem{lemma}[theorem]{Lemma}

\newtheorem{corollary}[theorem]{Corollary}

\newtheorem{definition}{Definition}

\newcommand{\order}{{\mathcal O}}
\newcommand{\oi}{O}
\newcommand{\bbribe}{\hat{{\order}}}
\newcommand{\sbribe}{\hat{{\oi}}}

\newcommand{\seats}[2]{\text{\it seats}_{#2}^{#1}}

\newcounter{examplecounter}
\newenvironment{Example}
    {\refstepcounter{examplecounter}\par\noindent\textbf{Example \theexamplecounter.}\hspace{1pt}}
    {\hfill\qedsymbol\par}

\newtheorem*{rtheorem}{\theremindertheorem}

\newcommand{\theremindertheorem}{}

\newenvironment{reminder}[1]
  {\renewcommand{\theremindertheorem}{Reminder of Theorem #1}\begin{rtheorem}}
  {\end{rtheorem}}

\newenvironment{proofnoqed}[1][Proof]{%
  \par\noindent\textbf{#1: }\ignorespaces 
}{%
  \par\medskip 
}

\title{Bribery for Coalitions in Parliamentary Elections}

\author{
Hodaya Barr
\and
Yonatan Aumann\And
Sarit Kraus\\
\affiliations
Bar-Ilan University, Israel\\
\emails
odayaben@gmail.com, aumann@cs.biu.ac.il, sarit@cs.biu.ac.il
}

\begin{document}

\maketitle

\begin{abstract}
We study the computational complexity of bribery in parliamentary voting, in settings where the briber is (also) interested in the success of an entire set of political parties - a ``coalition'' - rather than an individual party. 
We introduce two variants of the problem: the Coalition-Bribery Problem (CB) and the Coalition-Bribery-with-Preferred-party Problem (CBP). In CB, the goal is to maximize the total number of seats held by a coalition,  while in CBP, there are two objectives: to maximize the votes for the preferred party, while also ensuring that the total number of seats held by the coalition is above the target support (e.g. majority).

We study the complexity of these bribery problems under two positional scoring functions - Plurality and Borda - and for multiple bribery types -  $1$-bribery, $\$$-bribery, swap-bribery, and coalition-shift-bribery. We also consider both the case where seats are only allotted to parties whose number of votes passes some minimum support level and the case with no such minimum.  We provide polynomial-time algorithms to solve some of these problems and prove that the others are NP-hard.
\end{abstract}

\section{Introduction}
The complexity of bribery in voting has been widely studied (e.g. \cite{faliszewski2006complexity,elkind2009swap,faliszewski2009hard,faliszewski2016control,tao2022hard,keller2018approximating,yang2020complexity,zhou2020parameterized,elkind2020algorithms}). In the classic bribery setting, voters have preferences over candidates, and a briber seeks to promote a \emph{specific} candidate by way of bribing voters to change their preference orders. In parliamentary elections, however, voters - and bribers - are often interested not only in one specific party but, no less, in a set, or \emph{coalition} of parties. Indeed, promoting a specific party may be of little significance if the party ends up outside the ruling coalition. As such, citizens and bribers may opt to promote parties other than their most favorite one  - if doing so would result in a ruling coalition that better aligns with their views. In the recent French elections, for example, right-wing parties won in the first round, but in the second round, the left-wing and centrist parties united in an alliance. Many voters strategically cast their ballots for candidates who were not their top preference, yet this tactic enabled them to achieve their goal of forming a left-center government \cite{toussaint2024thanks}. Accordingly, we introduce and study bribery in such settings, where the briber is interested (at least partially) in promoting a coalition (/set) of parties, rather than a single party. As we exhibit, bribery in such cases can become significantly more complex than in the single-party case.

We consider two variants of the problem. In the \emph{Coalition-Bribery Problem (CB)} the goal is to maximize the total number of seats of a given coalition of parties in parliament. 
In the Coalition-Bribery-with-Preferred-party Problem (CBP), the goal is to maximize the number of seats won by a specific preferred party within the coalition, while ensuring that the total number of seats of the coalition exceeds some predefined goal (e.g. majority). We consider these two problems in a variety of settings, both in terms of the bribery cost functions and in terms of voting procedures, as follows.

\textbf{Types of Bribery.}
Several types of bribery have been defined in the literature.
In this paper, we consider four: $1$-bribery, $\$$-bribery, swap-bribery, and coalition-shift-bribery (see Section \ref{sec:prelim} for a full explanation).
We note that coalition-shift-bribery is an extension of the model of shift-bribery~\cite{elkind2009swap} to the coalitional setting.  

\textbf{Voting rules and Seat Allocation.}
We consider Plurality and Borda scoring rules.  We also consider both the case where seats are only allotted to parties that pass some minimum support threshold and the case with no such minimum~\cite{slinko2010proportional,put2016complexity}.

\begin{table*}[t]
    \setlength{\belowcaptionskip}{0pt} 
    \centering
    \begin{tabular}{|c|c|c|c|c|c|}
        \hline
        Bribery type & $1$ & $\$$  & Swap & Coalition-Shift\\
        \hline
         Plurality$_{0}$-CB & P (C\ref{cor:dolar-Plurality-t-CBP-plurality-CBP}) & P (C\ref{cor:dolar-Plurality-t-CBP-plurality-CBP}) &   P (C\ref{cor:Plurality-cbp}) & P (C\ref{cor:Plurality-cbp}) \\
         Plurality$_t$-CB & P (C\ref{cor:dolar-Plurality-t-CBP-plurality-CBP}) & P (C\ref{cor:dolar-Plurality-t-CBP-plurality-CBP})  & NP-hard (C\ref{cor:swap-Plurality-t-CBP})  & NP-hard (T\ref{thm:shift-Plurality-t-CBP})\\
         Plurality$_{0}$-CBP &  P (C\ref{cor:dolar-Plurality-t-CBP-plurality-CBP}) &  P (C\ref{cor:dolar-Plurality-t-CBP-plurality-CBP}) & P (T\ref{thm:Plurality-ccbp}) &  P (T\ref{thm:Plurality-ccbp})\\
         Plurality$_t$-CBP & P (C\ref{cor:dolar-Plurality-t-CBP-plurality-CBP}) & P (T\ref{thm:dolar-Plurality-t-CBPP}) & NP-hard (C\ref{cor:swap-non-shift-Plurality-t-CBPP}) & NP-hard (C\ref{cor:swap-non-shift-Plurality-t-CBPP})\\
        \hline
        Borda$_{0}$-CB & P (C\ref{cor:1-dolar-shift-Borda-CBP}) & P (C\ref{cor:1-dolar-shift-Borda-CBP})& NP-hard (T\ref{thm:swap-borda-cbp})  & P (C\ref{cor:1-dolar-shift-Borda-CBP}) \\ 
        Borda$_t$-CB  &  NP-hard (T\ref{thm:1-Borda-t-cbp})  & NP-hard (T\ref{thm:1-Borda-t-cbp})  & NP-hard (C\ref{cor:shift_swap_non-Borda-t-cbp})$^*$  &  NP-hard (C\ref{cor:shift_swap_non-Borda-t-cbp})$^*$ \\
        Borda$_{0}$-CBP  & P (T\ref{thm:dolar-Borda-CBPP}) & P (T\ref{thm:dolar-Borda-CBPP}) & NP-hard (T\ref{thm:swap-borda-cbp}) & P (T\ref{thm:dolar-Borda-CBPP}) \\
        Borda$_t$-CBP  & NP-hard (C\ref{cor:all-Borda-t-ccbp})  & NP-hard (C\ref{cor:all-Borda-t-ccbp})   & NP-hard (C\ref{cor:all-Borda-t-ccbp})$^*$ &  NP-hard (C\ref{cor:all-Borda-t-ccbp})$^*$\\
        \hline
    \end{tabular}
    \caption{Summary of complexity results. A $t$ subscript denotes the variant with a minimum support threshold, and 0 - the variant without. The parentheses refer to the corresponding Theorem(T) or Corollary(C) in the text. A $^*$ indicates that the result also follows from \protect\cite{put2016complexity}.}
    \label{tab:sum_complexity}
\end{table*}
\textbf{Results.}
A summary of the results is presented in Table~\ref{tab:sum_complexity}. Notably, we show that for Plurality all four variants (CB and CBP, with and without minimum support requirement) remain polynomial for $1$-bribery and $\$$-bribery even when adding coalition considerations.  This extends to the coalitional setting the known results from the single party setting, though the algorithms are, at times, considerably more complex.  On the other hand, for swap-bribery and shift-bribery the problem becomes NP-hard even for Plurality when there is a minimum support requirement (for both CB and CBP).  The same problems are polynomial in the corresponding single party setting~\cite{put2016complexity}.  For Borda scoring, we show that without a minimum support requirement  $1$-bribery, $\$$-bribery, and shift-bribery are solvable in polynomial time, while swap-bribery is NP-hard.  Once a minimum support requirement is introduced - all variants become NP-hard.  

\paragraph{Proofs.} 
Many proofs are deferred to the appendix, including most of the harness ones.  

\section{Related Work}

Several problems in the literature study methods to alter the outcome of an election, such as control (e.g \cite{put2016complexity,bartholdi1992hard}) and manipulation (e.g \cite{brelsford2008approximability,faliszewski2015complexity}).
In this paper, we focus on the bribery problem, introduced by \cite{faliszewski2006complexity}.
Various bribery price functions have been considered, and the complexity of the bribery problem depends on the type of bribery and the election's voting procedure. These research questions have been extensively studied (e.g.~\cite{bartholdi1992hard,tao2022hard,maushagen2022complexity,elkind2009swap,brelsford2008approximability,faliszewski2008Nonuniform,faliszewski2015complexity,parkes2012complexity}). See also~\cite{faliszewski2016control} for a survey.

\textbf{Complexity of the Bribery Problem.} 
In the case of  Plurality it has been shown that the bribery problem is solvable in polynomial time for several types of bribery: $1$-bribery, $\$$-bribery \cite{faliszewski2006complexity},  non-uniform-bribery \cite{faliszewski2008Nonuniform} and swap-bribery \cite{elkind2009swap}.
For Borda, on the other hand, the bribery problem is NP-hard even for $1$-bribery~\cite{brelsford2008approximability}, and NP-completeness has been proven for other types of bribery, such as shift-bribery~\cite{elkind2009swap} by provide an NP algorithm.
 

\textbf{Parliamentary Elections.}
Bribery in parliamentary elections was studied in~\cite{put2016complexity}.  They show that shift bribery is solvable in polynomial time under Plurality, for both the case with a minimum support threshold and without. 
For Borda scoring, they prove that the shift bribery problem is in P if there is no minimum support requirement, but is NP-hard with such a requirement. \cite{put2016complexity} however, only consider bribery for promoting a single party, and not sets of parties - which is the focus of our research. 

\textbf{Power Indexes in Weighted Voting Games.}
The research on manipulating power indexes (e.g. Shapley–Shubik or Banzhaf) in Weighted Voting Games (WVG) is also related, as the parliamentary process with parties and a quota is a WVG. \cite{Aziz2011False,rey2014false} study splitting and merging parties to increase the power indexes.   \cite{rey2018structural} study adding or omitting players to increase or decrease the power indexes.  All these problems are shown to be computationally hard. \cite{zick2011shapley} and \cite{zuckerman2012manipulating} study affecting the power index through changes in the quota.  Our work differs from this line in several ways, as we do not consider the power indexes, we consider voter bribery - not party structure changes, and we consider minimum support thresholds.        

\textbf{Strategic Voting and Coalition Preferences.}
Cox~\cite{cox2018portfolio} explores portfolio-maximizing strategic voting, demonstrating that voters often base their decisions not only on their preferences for an individual party but also on their intent to influence coalition outcomes and maximize governmental influence.
The topics of strategic and tactical coalition voting have been explored in numerous studies (e.g. \cite{meffert2010strategic,bowler2010strategic,gschwend2016drives,MCCUEN2010316}), examining how voters may choose to vote based on coalition considerations, even if that means not voting for their most preferred party.

\section{Model}
\label{sec:prelim}
A parliamentary election is a triplet $E=(C, V, {\mathcal O})$, where $C = \{c_1, \ldots, c_m\}$ is the set of political parties, $V= \{v_1,\ldots,v_n\}$ is the set of voters, and $\order=(\oi_1, \ldots, \oi_n)$ is a sequence of preference orders over the parties, where $\oi_i$ is the preference order of the voter $v_i$. 
We denote by $pos(O_i,c_j)$ the position of party $c_j$ in $\oi_i$.

The voters' preference orders determine the number of seats in parliament allotted to each party. In this work we focus on the \emph{fractions} of seats allotted to the parties, not the nominal number, and following~\cite{put2016complexity}, we assume that any fraction is possible, even if it represents a fractional number of parliament seats.  The fractions are determined using \emph{positional scoring functions}. 

A \textbf{positional scoring function} is a function, $\gamma(O_i,c_j)$, that assigns ``points'' according to rank; that is, order $O_i$ ``gives'' $\gamma(O_i,c_j)$ points to  party $c_j$. Extending the function to sets, we define $\gamma(\order',C')=\sum_{O_i\in \order'}\sum_{c_j\in C'}\gamma(O_i,c_j)$. 

We consider two standard positional scoring functions:
\begin{enumerate}
    \item \textbf{Plurality:} $\mbox{\it Plurality}(O_i,c_j)=1$ if $pos(O_i,c_j)=1$ and $0$ otherwise.  
    \item \textbf{Borda:} ${Borda}(O_i,c_j) = m-pos(O_i,c_j)$, where $m$ is the total number of parties. 
\end{enumerate}
The fraction of seats allotted to each party is proportional to its number of points.  However, in order to avoid small parties, many parliamentary systems introduce a \emph{min size threshold} $t$, where parties that got less than a $t$ fraction of the total number of points are eliminated.  In this case  - where $\gamma(\order,c_j)< t\cdot\gamma(\order,C)$ - we say that $c_j$ is \emph{inactive (under $\order$)}, otherwise, it is \emph{active}. The votes for an active (inactive) party are said to be \emph{active} (res. \emph{inactive}).  In all, 
for positional scoring function $\gamma$ and min-size threshold $t$, the fraction of seats allotted to party $c_j$ is:
\vspace{-3ex}
\begin{align*}
    \seats{\gamma}{t}(\order,c_j)=\begin{cases}
        \frac{\gamma(\order,c_j)}{\sum_{c_{j'} \text{ is active }} \gamma(\order,c_{j'})} & c_j \text{ is active} \\
        0 & \text{otherwise}
    \end{cases}
\end{align*}
The function $\seats{\gamma}{t}$ is naturally extended to sets of parties: for $C'\subseteq C, \seats{\gamma}{t}(\order,C')=\sum_{c\in C'}\seats{\gamma}{t}(\order,c)$.

The special $t=0$ is, at times, simpler and will be treated separately.

\subsubsection*{Bribery}
Bribery allows for the modification of the voters' preference orders.  Technically, we equate a bribe with the resultant sequence of modified orders $\bbribe=(\sbribe_1,\ldots,\sbribe_m)$. The cost of changing voter $v_i$'s preference order from $O_i$ to $\sbribe_i$ is $\pi_i(\oi_i,\sbribe_i)$. The cost of the entire bribe is $\pi(\order,\bbribe)=\sum_{i=1}^n\pi_i(\oi_i,\sbribe_i)$.
For brevity, when clear from the context, we omit the original order $\sbribe_i$, writing $\pi_i(\hat{O}_i)$.

The literature considers several different bribery types, distinguished both by the permissible changes to voters' preferences and by the cost to do so. In this paper, we consider the following bribery types:
\begin{description}
    \item[$1$-bribery:] All changes to voters' preference orders are permissible. The cost function is $\pi^{1}_i(O_i,\sbribe_i)\equiv 1$, for all bribes and all voters.   
    \item[$\$$-bribery:]  All changes are permissible. For each $i$, there is a value $p_i$, with $\pi^{\$}_i(O_i,\sbribe_i)\equiv p_i$, for any $\sbribe_i$~\cite{faliszewski2006complexity}.
    \item[Swap-bribery:] All changes are permissible.  The cost depends on pairs of parties that swapped order between the preference orders~\cite{elkind2009swap}. Specifically, for parties $x,y$ there is a cost $sw_i(x,y)$ for moving $y$ from a location below $x$ to one above $x$. The total cost for $\sbribe_i$ is the aggregate \sloppy{$\pi^{swap}_i(O_i,\sbribe_i)=\sum_{x,y: \text{$y$ is below $x$ in $O_i$ and above $x$ in $\sbribe_i$}}sw_i(x,y)$}. 
    \item[Shift-bribery:] In the classic bribery setting, shift-bribery assumes one \emph{preferred candidate} $c_1$, and only its location can be shifted, 
    and only upward. The cost is a function of the number of locations the preferred candidate is shifted, $\pi^{shift}_i(O_i,\hat{O}_i)=s_i(|\{ c_j : \text{the order of $c_1$ and $c_j$ is different in $O_i$ and $\sbribe_i$}\}|)$ for some monotone function $s_i$~\cite{elkind2009swap}. 

    Adapting the definition to the coalitional setting, in \textbf{coalition-shift-bribery} there is a \emph{preferred set} $A$ of parties (the coalition), and it is possible to shift the location of any $a\in A$ upwards (that is, for $x\in A$ and $y\not\in A$, it is forbidden that $x$ was above $y$ in $O_i$ but is below $y$ in $\hat{O}_i$).  
    The cost is 
    $\pi^{shift}_i(O_i,\hat{O}_i)=s_i(|\{ (x,y) : \text{$y$ is below $x$ in $O_i$ and above $x$ in $\sbribe_i$} \}|)$ for some monotone function $s_i$.
\end{description}

\subsection{Problem Definition}
We now provide the formal definitions of two problem variants.
In the first - CB - there is a subset $A$ of parties - the \emph{coalition} - which the briber aims to boost their total support.  In the second - CBP - there is, in addition, a \emph{preferred party} within the coalition, which the briber aims to boost its support relative to the coalition's overall support.        
\begin{definition}[Coalition-Bribery-Problem (CB)]~\newline
\underline{Given:} an election $E= (C,V,O)$, a coalition of parties $A\subseteq C$, a seat allocation function $\seats{\gamma}{t}$, a bribery type with price function $\pi$, 
 a budget $B$, and target support $\varphi \in [0,1]$, \newline
 \underline{Decide:} whether there exists a bribe $\bbribe$ such that 
(a) $\pi(\order,\bbribe) \leq B$ (does not exceed the budget), 
(b) $\seats{\gamma}{t}(\bbribe,A) \geq \varphi\cdot\seats{\gamma}{t}(\bbribe,C)$.
\end{definition}
\begin{definition}[Coalition-Bribery-with-Preferred-party-Problem (CBP)]~\newline
\underline{Given:} as in CB, and, in addition: (i) a \emph{preferred party} $c_{j_1}\in A$, and (ii) a target ratio $\rho \in [0,1]$, \newline
\underline{Decide:} whether there exists a bribe $\bbribe$ for which (a) and (b) hold, and in addition:
(c) $\seats{\gamma}{t}(\bbribe,c_{j_1})\geq \rho\cdot  \seats{\gamma}{t}(\bbribe,A)$ (the preferred party gets at least $\rho$ percentage among the coalition $A$).
\end{definition}
\paragraph{Notations.}
Without loss of generality, we assume that the preferred party (in CBP) is $c_1$, and $A = \{c_1,\ldots,c_h\}$.   We denote $A_{-1}=A\setminus\{c_1\}$ and $\bar{A}=C\setminus A$.  Also, $top(O_i)$ denotes the top party in $O_i$.

For brevity, we denote the CB problems associated with the functions $\seats{\text{Plurality}}{t}$ and $\seats{\text{Borda}}{t}$ by Plurality$_t$-CB and Borda$_t$-CB, respectively, and analogously, Plurality$_t$-CBP and Borda$_t$-CBP. 

\textbf{Examples.}
Consider a setting with 3 parties: X, Y, and Z.  The briber's preferred party is X, and the coalition of interest is $A=\{ \text{X, Y}\}$. The voting rule is plurality, and the target support for the coalition is $\varphi = 0.5$. Suppose there are 100 voters: 35 voting  X (that is, X is their top preference), 15 to Y, and 50 to Z. 

If there is no minimum size threshold, then all votes are active, the coalition gets 50\% of the seats, and no bribe is necessary. Suppose that the minimum size threshold is $t=20\%$. Then, with no bribe the votes to Y are inactive, and the coalition $A$ gets only a $35/(35+45)\approx 41\%$ fraction of the seats. Under 1-bribery, the optimal bribe to affect the target support is to bribe 5 voters of Z to vote for Y instead. This activates the votes of Y, for a total of $35+20=55$ for the coalition. 

Next, suppose the bribery type is \$-bribery, and bribing voters of Z costs \$2, while bribing those of X costs \$1.  Then, the least cost bribe to bring $A$ to the support level $\rho=0.5$ is to bribe 5 voters of X and have them vote for Y, in which case $A$  gets exactly 50\% of the seats. 
Finally, suppose that we are in the CBP setting, where there is both a target support $\varphi=0.5$ for the coalition, and a target ratio $\rho=61\%$ for X within the coalition. Then, transferring 5 seats from X to Y will bring the ratio of X within the coalition down to 60\%$<\rho$.  So, the optimal bribe is to bribe 3 voters of X and 2 voters of Z, all to vote for Y.  The seat allocation is then: 32 for X, 20 for Y, and 43 for Z.  The cost is $3\times \$1+2\times\$2=\$7$.

\section{Plurality}
\subsection{1-Bribery and \$-Bribery}
We first provide a polynomial algorithm that solves the Plurality$_t$-CBP with \$-bribery.  
The algorithm splits the bribing process into two: first considering the set of voters to bribe - and omitting their current votes from the pool of votes - and only afterward determining their new votes.   

Accordingly, we define the notion of an \textit{undetermined bribe}, which is the process wherein a set of voters is bribed, thus eliminating their original votes, but their new votes are not yet determined.  Technically, an undetermined-bribe is a set $\hat{V} =\{ v_{i_1},\ldots,v_{i_{\ell}}\}$ of (bribed) voters.  The cost of the undetermined bribe $\hat{V}$ is $\pi^{\$}(\hat{V})=\sum_{i:v_i\in \hat{V}}p_i$.  Note that under \$-bribery this cost suffices to later determine the votes of all members of $\hat{V}$ in any way one may wish. 

The key ingredient in the algorithm is a dynamic programming computation of a function $g$, such that $g(\ell, a_{\bar{A}},d,a_{A_{-1}})$ is the least cost of an undetermined-bribe $\hat{V}$ that satisfies all the following requirements:
\begin{itemize}
    \item $|\hat{V}| = \ell$,
    \item $top(O_i)\neq c_1$, for all $v_i\in \hat{V}$, 
    \item from the votes of $V\setminus \hat{V}$ (the not bribed voters), there are $a_{\bar{A}}$ active votes for parties of $\bar{A}$. 
    \item there is a way to distribute $d$ additional votes to parties of $A_{-1}$ in a way that, together with the votes of $V\setminus \hat{V}$, the total number of active votes for parties of $A_{-1}$ is $a_{A_{-1}}$.  
\end{itemize}
If there is no undetermined bribe that satisfies the above, then $g(\ell, a_{\bar{A}},d,a_{A_{-1}})=\infty$.

Let $mincost(c_j,\ell)$ be the least cost of bribing $\ell$ voters who in $\order$ voted for $c_j$ (with $\infty$ if there is no such bribe).  
\begin{lemma}
It is possible to compute $g(\ell, a_{\bar{A}},d, a_{A_{-1}})$, for all $0\leq \ell, a_{A_{-1}}, d, a_{\bar{A}}\leq n$, and $mincost(c_j,k)$, for $1\leq j\leq m, 1\leq \ell \leq n$ in polynomial time.
    \label{alg:f-i-j-k-c-Plurality-t-CBPP}
\end{lemma}
\begin{proofnoqed}
For party $c_j$, let $V_j=\{ v_i : top(O_i)=c_j\}$ - the set of voters who (before bribes) vote for $c_j$.  For a set of parties $D$, set $V_D=\cup_{c_j\in D}V_j$. 
For computing $mincost(c_j,\ell)$ it suffices to first sort the set $\{ p_i : v_i\in V_j\}$.  Then $mincost(c_j,\ell)$ is the sum of the $\ell$ least values in this set.

For $g$:
When seeking the least-cost undetermined bribe $\hat{V}$, we separately compute $\hat{V}_{A_{-1}}$ - the bribed voters who initially voted for $A_{-1}$ - and $\hat{V}_{\bar{A}}$ - the bribed voters who initially voted for $\bar{A}$.

\paragraph{Computing $\hat{V}_{\bar{A}}$.}
For a set of parties $D\subseteq \bar{A}$, let $f(D,\ell,a_D)$ be the least cost of an undetermined bribe $\hat{W}\subseteq V_D$ such that: (i) $|\hat{W}|=\ell$, (ii) based on the votes of $V_D\setminus \hat{W}$ exactly $a_D$ votes are active.  

We provide a dynamic programming process for computing $f(D,\ell,a_{\bar{A}})$ for all $D\subseteq \bar{A}$, and $0\leq \ell,a_{\bar{A}}\leq n$.  First, for $D$ of size 1:
\begin{align}
    f(\{j\},\ell,a_{\{ j\}})= mincost(c_j,\ell) \nonumber
\end{align}
if either $a_{\{j\}}=|V_j|-\ell\geq t$ or $|V_j|-\ell< t$ and $a_{\{j\}}=0$.  In all other cases, $f(\{j\},\ell,a_{\{ j\}})=\infty$. 

Now, 
\begin{align*}
    f(D\cup \{j\},\ell,a_{D\cup\{j\}})=
\;\; \min_{0\leq \ell'\leq \ell, 0\leq a_{D} \leq a_{D\cup \{ j\}}} \\
\{     f(D,\ell',a_D) +     f(\{ j\},\ell-\ell',a_{D\cup\{ j\}}-a_D) \} 
\end{align*}
This allows to iteratively compute $f(\bar{A},\ell,a_{\bar{A}})$, for all $\ell,a_{\bar{A}}$, by adding parties one by one.

\textbf{Computing $\hat{V}_{A_{-1}}$.}
For a set of parties  $D\subseteq A_{-1}$, let $h(D,\ell,d , a_{D})$ be the least cost of an undetermined bribe $\hat{W}\subseteq V_D$ such that: (i) $|\hat{W}|=\ell$, (ii) with $d$ additional votes, it is possible to obtain a total of $a_{D}$ active votes for parties of $D$.

Again, dynamic programming allows computing
$h(A_{-1},\ell,d,a_{A_{-1}})$ for all $0\leq \ell,d, a_{\bar{A}}\leq n$.  First, for $D$ of size 1, 
\begin{align*}
    h(\{j\},\ell,d,a_{\{ j\}})=mincost(c_j,\ell)
\end{align*}
if either one of the following holds: 
    \begin{itemize}
        \item $|V_j|-\ell+d\geq t$ and $a_{\{ j\}}=|V_j|-\ell+d$,
        \item $|V_j|-\ell+d< t$ and $a_{\{ j\}}=0$
    \end{itemize}
In all other cases, $h(\{j\},\ell,d, a_{\{ j\}})=\infty$.

Now, 
\begin{align*}
    h(D\cup \{j\},\ell,d,a_{D\cup\{j\}})=
\;\; \min_{0\leq \ell'\leq \ell, 0\leq d'\leq d, 0\leq a_{D} \leq a_{D\cup\{ j\}}}\\ \{     h(D,\ell',d',a_D) +  h(\{ j\},\ell-\ell',d-d',a_{D\cup\{ j\}}-a_D) \} 
\end{align*}
This allows to iteratively compute $h(A_{-1},\ell,d,a_{A_{-1}})$, for all $\ell,d,a_{A_{-1}}$, by adding parties one by one.

\paragraph{Computing $g$.}
Now, having computed $h(A_{-1},\ell,d,a_{A_{-1}})$ and $f(\bar{A},\ell,a_{\bar{A}})$ for all $\ell, d,a_{A_{-1}},a_{\bar{A}}$, we compute $g$:
\begin{align*}
    g(\ell,a_{\bar{A}},d,a_{A_{-1}}) =& \\
\;\; \min_{0\leq \ell'\leq \ell} \{  h(A_{-1},\ell',d,a_{A_{-1}})
+&f(\bar{A},\ell-\ell',a_{\bar{A}})\} \;\;\;\;\;\;\;\; \qed
\end{align*}
\end{proofnoqed}
\begin{algorithm}
    \begin{algorithmic}[1]
        \FORALL{$0\leq \ell, a_{\bar{A}},d, a_{A_{-1}} \leq n$}
            \STATE $B' = g(\ell, a_{\bar{A}},d, a_{A_{-1}})$
            \IF{$B' > B$}
                \STATE Continue
            \ENDIF  
            \IF{$d> \ell$}
                \IF{$mincost(c_1, d-\ell) > B-B'$}
                    \STATE continue
                \ENDIF
            \ENDIF
            \STATE $a_{c_1}=\gamma(\mathcal{O}, c_1)+\ell-d$
            \IF{($a_{c_1} < T$) }
                \STATE $a_{c_1}=0$
            \ENDIF
            \IF{$\frac{a_{A_{-1}}+a(c_1)}{a_{A_{-1}}+a(c_1) + a_{\bar{A}}} \geq \varphi \text{ and }  \frac{a(c_1)}{a_{A_{-1}}+a(c_1)}\geq \rho$}
                \RETURN The budget is sufficient
            \ENDIF
        \ENDFOR
    \RETURN The budget is insufficient
    \end{algorithmic}
    \caption{Algorithm for Plurality$_t$-CBP with the $\$$-bribery}
    \label{alg:dolar-Plurality-t-CBPP}
\end{algorithm}
Given the efficient computation of $g$, Algorithm \ref{alg:dolar-Plurality-t-CBPP} solves Plurality$_{t}$-CBP by iterating over all possible combinations of $\ell, a_{\bar{A}},d, a_{A_{-1}}$, and for each, performing three steps:

\textbf{First step.}\textit{
Compute $g(\ell, a_{\bar{A}},d, a_{A_{-1}})$ and check if the budget suffices (lines 2-5).}

\textbf{Second step.}
\textit{Add $d$ bribed votes to support $A_{-1}$ and the rest to $c_1$ (lines 5-8).}
The bribe of $g$ provides $\ell$ ``free'' votes, while we are seeking to add $d$ votes for $A_{-1}$.  If $d>\ell$, then we need to add $d-\ell$ votes from $c_1$.  Lines 5-7 check if this is possible. Line 8 adds the remaining $\ell-d$ (which may be negative) to the support of $c_1$.

\textbf{Third step.} 
\textit{Check if the bribe is successful (lines 9-13)}
Lines 9-10 adjust the seats of $c_1$ according to the threshold. Finally, line 11 checks whether the bribe satisfies the required target support and target ratio ($\varphi$ and $\rho$). 
If yes, line 12 returns that the budget is sufficient and terminates. Otherwise, the algorithm continues to the next iteration.
If no valid bribe is found in any iteration, the algorithm returns that the budget is insufficient (line 13).  

We obtain:
\begin{theorem}
     Algorithm \ref{alg:dolar-Plurality-t-CBPP} solves Plurality$_t$-CBP for $\$$-bribery, in polynomial time.
     \label{thm:dolar-Plurality-t-CBPP}
\end{theorem}
Since $1$-bribery is a special case of $\$$-bribery, with $p_i = 1$ for all voters, and Plurality$_t$-CB is a special case of Plurality$_t$-CBP obtained by setting $\rho = 0$, we obtain:
\begin{corollary}
    Plurality$_t$-CB can be solved in polynomial time for $1$-bribery and $\$$-bribery.
    \label{cor:dolar-Plurality-t-CBP-plurality-CBP}
\end{corollary}

\subsection{Swap-Bribery and Coalition-Shift-Bribery without Threshold}
We show that Plurality$_{0}$-CBP can be solved in polynomial time by solving $O(n^2)$ instances of the minimum-cost flow (MCF) problem, defined as follows.
\begin{definition}[Minimum-Cost Flow (MCF)]
    Given a directed graph $G=(U, D)$, source and sink nodes $s,t\in U$,
    non-negative capacity and cost functions $cap(e)$ $cost(e)$ for each $e \in D$, and demand $d$, 
    the goal is to find a flow $f: D \rightarrow \mathbb{R}$  that minimizes the weighted cost $\sum_{e\in D}f(e)cost(e)$, while satisfying: (i) capacity constraint: $0\leq f(e)\leq cap(e)$ for all $e\in D$, and (ii) flow conservation: 
    $\sum_{u\in V}f(u,v)= \sum_{u\in V}f(v,u)$, for $v\in V\setminus \{ s,t\}$, (iii) demand fulfilled: $\sum_{u\in V}f(s,u)= \sum_{u\in V}f(u,t)=d$.  
\end{definition}
The MCF problem can be solved in polynomial time (see e.g.~\cite{ahuja1993network}).

We show how, given $k_1,k_{A_{-1}}$, to construct an MCF instance $M_{k_1,k_{A_{-1}}}$ that admits a solution with weighted cost $B'$ iff there exists a bribe $\bbribe$ with cost $B'$ such that $Plurality(\bbribe,c_1)=k_1$ and $Plurality(\bbribe,A_{-1})=k_{A_{-1}}$. 

For each $i$, let $\hat{O}_{i,c_1}$ be the bribe of $v_i$ that brings $c_1$ to the top position (without moving any other parties).  Similarly, let $\hat{O}_{i,A_{-1}}$ 
and $\hat{O}_{i,\bar{A}}$ are the minimal cost bribes of $v_i$ that bring a member of $A_{-1}$ and of $\bar{A}$ to the first position, respectively.  Note that $\hat{O}_{i,\cdot}$ may be $O_i$, if no bribe is necessary. Also, for determining $\hat{O}_{i,A_{-1}}$ and $\hat{O}_{i,\bar{A}}$ we only need to consider bribes that bring a single party to the top position, with no other changes. So, all these bribes, together with their costs, can be computed in polynomial time for both Swap-Bribery and Coalition-Shift-Bribery.

The graph, capacities, and costs in $M_{k_1,k_2}$ are defined as follows (see Figure \ref{fig:Plurality-ccbp}.  The graph nodes are:  
\begin{itemize}
    \item One source node $s$ and one target node $t$.
    \item three nodes $a_{1},a_{A_{-1}}$ and $a_{\bar{A}}$.
    \item for each voter $v_i$ four nodes: $v_i,u_{i,1}, u_{i,A_{-1}}$ and $u_{i,\bar{A}}$.
\end{itemize} 
The edges, capacities, costs, and demand are:
\begin{itemize}
    \item edges $(s,a_{1}), (s,a_{A_{-1}}), (s,a_{\bar{A}})$, all with cost 0 and capacities  $cap(s,a_{1})=k_1,cap(s,a_{A_{-1}})=k_2,cap(s,a_{\bar{A}})=n-k_1-k_2$.  
    \item for each $i$, edges $(a_{1},u_{i,1}), (a_{A_{-1}},u_{i,A_{-1}})$ and $(a_{\bar{A}},u_{i,\bar{A}})$, all with capacity $1$ and cost $0$.
    \item for each $i$,  edges $(u_{i,\alpha},v_i)$ for $\alpha\in \{ 1,A_{-1},\bar{A}\}$, with $cost(u_{i,\alpha},v_i)=\pi_i(O_i,\hat{O}_{i,\alpha})$, and capacity 1.
    \item for each $i$,  edge $(v_i,t)$ with cost 0 and capacity 1.
    \item the demand is $d=n$.
\end{itemize}
\begin{figure}[tbhp]
    \centering
    \vspace{-0ex}
    \includegraphics[width=\linewidth]{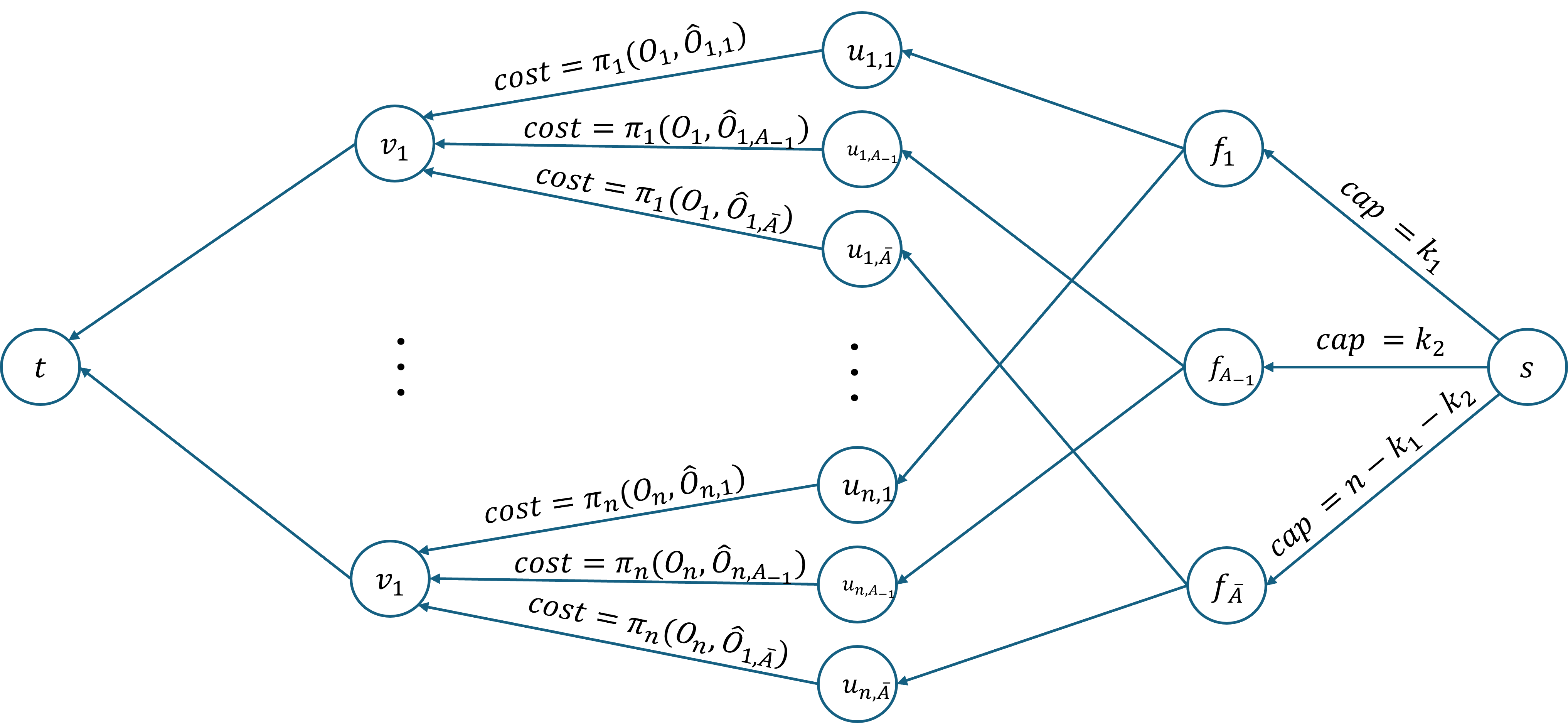}
    \caption{Network for proof of Theorem \ref{thm:Plurality-ccbp}. Unless otherwise stated, the default values are $cost = 0$ and $cap = 1$}
   
    \label{fig:Plurality-ccbp}
\end{figure}
\begin{lemma}
\label{lem:MCF}
    $M_{k_1,k_{A_{-1}}}$  admits a solution with weighted cost $B'$ iff there exists a bribe $\bbribe$ with cost $B'$ such that $Plurality(\bbribe,c_1)=k_1$ and $Plurality(\bbribe,A_{-1})=k_{A_{-1}}$.
\end{lemma}
\begin{proof}
Let $f$ be a flow that solves $M_{k_1,k_{A_{-1}}}$ with cost $B'$.  We construct a bribe $\bbribe$ with the same costs as claimed by the lemma. Since all inputs are integral we may assume that $f$ is integral.  Since the demand is fulfilled, and $cap(v_i,t)=1$, it must be that $f(v_i,t)=1$ for all $i$.  Hence, for each $i$, there exists exactly one $\alpha_i\in \{ 1,A_{-1},\bar{A} \}$ with $f(u_{i,\alpha_i},v_i)=1$.  So, define $\hat{O}_i=\hat{O}_{i,\alpha_i}$. By construction, the cost of this bribe is $\pi_i(O_i,\hat{O}_{i,\alpha_i})$, and the total cost of the entire bribe is $\sum_{i=1}^{n}\pi_i(O_i,\hat{O}_{i,\alpha_i})$ - which is also the weighted cost the flow $f$. Now, since the demand $d=n$ is fulfilled, $f(s,a_1)=k_1$ and $f(s,A_{-1})=k_{A_{-1}}$. So, there are exactly $k_1$ (res. $k_{A_{-1}}$) indexes $i$, with $f(a_{1},u_{i,1})=1$ (res. $f(a_{A_{-1}},u_{i,A_{-1}})=1$).  So, $Plurality(\bbribe,c_1)=k_1$ and $Plurality(\bbribe,A_{-1})=k_{A_{-1}}$.

The proof of the other direction is similar.
\end{proof}
\begin{theorem}
    There is a polynomial-time algorithm that solves Plurality$_{0}$-CBP for swap-bribery and coalition-shift-bribery.
    \label{thm:Plurality-ccbp}
\end{theorem}
\begin{proof}
For each $0\leq k_1, k_{A_{-1}}\leq n$, $k_1+k_{A_{-1}} \leq n$, such that $\frac{k_1+k_2}{n} \geq \varphi$ and $\frac{k_1}{k_1+k_2} \geq \rho$, construct the MFC instance $M_{k_1,k_2}$.  Solve $M_{k_1,k_2}$.  If the solution is with cost $\leq B$, then, by Lemma \ref{lem:MCF}, there is also a bribe $\bbribe$ with the same cost, which solves the associated CBP$_{0}$ instance. Hence, return ``success''. If none of the $M_{k_1,k_{A_{-1}}}$ admit a solution with cost $\leq B$, then there is no solution to the CBP$_{0}$ problem, and return ``failure''. The number of pairs $k_1,k_{A_{-1}}$ to check is $O(n^2)$. 
\end{proof}
Since Plurality$_{0}$-CB is a special case of Since Plurality$_{0}$-CBP we obtain:
\begin{corollary}
    Plurality$_{0}$-CB can be solved in polynomial time for swap-bribery and coalition-shift-bribery.
    \label{cor:Plurality-cbp}
\end{corollary}

\subsection{Swap-Bribery and Coalition-Shift-Bribery with Threshold}
We show that Plurality$_t$-CB is NP-hard for both swap-bribery and coalition-shift bribery.
This holds even when the price function for coalition-shift bribery is restricted to the number of shifts multiplied by a fixed value for each voter.

For coalition-shift-bribery, the proof is by reduction from the 3-4-Exact-Cover
~\cite{brelsford2008approximability}.

\begin{theorem}
     Plurality$_t$-CB with coalition-shift-bribery is NP-hard, even when the price function is such that all $s_i(x)$'s are multiplicative.
    \label{thm:shift-Plurality-t-CBP}
\end{theorem}
From Theorem \ref{thm:shift-Plurality-t-CBP} we obtain:
\begin{corollary}
     Plurality$_t$-CB is NP-hard for swap-bribery.
    \label{cor:swap-Plurality-t-CBP}
\end{corollary}
\begin{proof}
Given an instance of coalition shift bribery with multiplicative $s_i$'s, we reduce it to an instance of swap bribery. For voter $v_i$ and $x,y\in C$ define $sw_i(x,y)=(s_i)'$ for $y\in A$ (where $(s_i)'$ is the slope of $s_i$), and $sw_i(x,y)=B+1$ otherwise.  The latter provides that only allowed swaps will be considered, and the former that the cost remains the same.
\end{proof}
\begin{corollary}
    Plurality$_t$-CBP is an NP-hard problem for swap-bribery and coalition-shift-bribery.
    \label{cor:swap-non-shift-Plurality-t-CBPP}
\end{corollary}

\section{Borda}
\subsection{Borda$_{0}$}
We show that Borda$_{0}$-CB and Borda$_{0}$-CBP for $1$-bribery, $\$$-bribery and coalition-shift-bribery can be solved in polynomial time.  Again, the algorithm uses dynamic programming.  We start with the following lemma.
\begin{lemma}
For voter $v_i$, let $h(i,k_{A_{-1}},k_1)$ be the least cost of a bribe $\hat{O}_i$ such that $\text{Borda}(\hat{O}_i,A_{-1})=k_{A_{-1}}$ and $\text{Borda}(\hat{O}_i,c_1)=k_1$, if such a bribe exists and $\infty$ otherwise.  For $1$-bribery, $\$$-bribery and coalition-shift-bribery, 
$h(i,k_{A_{-1}},k_1)$ can be computed in polynomial time, for any $i,k_{A_{-1}},k_1$.
\end{lemma}
\begin{proof}
{\bf {Computing $h$ for 1-bribery and \$-bribery.}}
If $O_i$ already gives the necessary supports then $h(i,k_{A_{-1}},k_1)=0$.  Otherwise, we say that $(k_{A_{-1}},k_1)$ is \emph{attainable} if there exists $\hat{O}_i$ with $\text{Borda}(\hat{O}_i,A_{-1})=k_{A_{-1}}$ and $\text{Borda}(\hat{O}_i,c_1)=k_1$.  If $(k_{A_{-1}},k_1)$ is attainable, then $h(i,k_{A_{-1}},k_1)=p_i$.  If $(k_{A_{-1}},k_1)$ is unattainable, then $h(i,k_{A_{-1}},k_1)=\infty$.  For 1-bribery, $p_i=1$. 

It remains to decide if $(k_{A_{-1}},k_1)$ is attainable. Suppose that $(k_{A_{-1}},k_1)$ is attainable by $\hat{O}_i$. Then, in particular, $pos(\hat{O}_i,c_1)=m-k_1$.  The members of $A_{-1}$ are placed below and above $c_1$, with $A^{\downarrow}$ below and $A^{\uparrow}$ above, and $A_{-1}=A^{\downarrow}\cup A^{\uparrow}$. The least number of points provided by $A^{\downarrow}$ is $|A^{\downarrow}|\frac{|A^{\downarrow}|-1}{2}$, which occurs when the members of $A^{\downarrow}$ occupy all the bottom places.  The maximum number of points that $A^{\downarrow}$ can provide is $|A^{\downarrow}|(k_1-\frac{|A^{\downarrow}|+1}{2})$ - which occurs when the members of $A^{\downarrow}$ are located immediately below $c_1$.  By shifting members of $A^{\downarrow}$ one by one between these two extreme settings, all intermediate values are also possible.  Similarly, the possible values for $\text{Borda}(\hat{O}_i, A^{\uparrow})$ are all those between $|A^{\uparrow}|(k_1+\frac{|A^{\uparrow}|+1}{2})$ and $|A^{\uparrow}|(m-\frac{|A^{\uparrow}|+1}{2})$.  

So, in all, $(k_{A_{-1}},k_1)$ is attainable iff there exist $\ell^{\downarrow},\ell^{\uparrow}$ (representing the sizes of $A^{\downarrow},A^{\uparrow}$), such that:
\begin{align*}
\bullet &\; \ell^{\downarrow}+\ell^{\uparrow}=|A_{-1}|, \\
\bullet &\; \ell^{\downarrow}\leq k_1, \ell^{\downarrow}< m-k_1,\\
\bullet &\; \ell^{\downarrow}\frac{\ell^{\downarrow}-1}{2} +
\ell^{\uparrow}(k_1+\frac{\ell^{\uparrow}+1}{2}) 
\leq k_{A_{-1}} \leq \\
&\;  \;\;\;\;\;\;\ell^{\downarrow}(k_1-\frac{\ell^{\downarrow}+1}{2})+
\ell^{\uparrow}(m-\frac{\ell^{\uparrow}+1}{2})
\end{align*} 
For any $k_1,k_{A{-1}}$ these constraints can be checked in polynomial time.

\paragraph{Computing $h$ for coalition-shift-bribery.} 
Let $\hat{O}_i$ be a bribe the realizes $h(i,k_{A_{-1}})$.  We may assume that the internal order between members of $A_{-1}$ in $\hat{O}_i$ is the same as in $O_i$.  Otherwise, we can swap between any two that or not in order, and the cost would not increase.  
By definition, $pos(\hat{O}_i,c_1)=m-{k_1}$. 
The members of $A_{-1}$ are placed below and above $c_1$, with $A^{\uparrow}$ above and $A^{\downarrow}$ below. Set $\ell^{\uparrow}=|A^{\uparrow}|, \ell^{\downarrow}=|A^{\downarrow}|$.  Since shift-bribery does not permit moving parties down, $pos(O_i,c_j)\leq m-k_1$, for $c_j\in A^{\downarrow}$.  We separately analyze the cost and points offered by $A^{\downarrow},c_1,$ and $A^{\uparrow}$.  

\paragraph{$A^{\downarrow}$:} The least cost and least points (for $A^{\downarrow}$) are afforded by the order wherein all members of $A^{\downarrow}$ are as in $O_i$.  In this case, the cost $0$, and the Borda points (provided $A^{\downarrow}$) are some value, which we denote $\underline{M}^{\downarrow}_{O_i}(\ell^{\downarrow})$.\footnote{$\underline{M}^{\downarrow}_{O_i}(\ell^{\downarrow})=\sum_{j=1}^{\ell^{\downarrow}}(m-pos(O_i,c_{h-\ell^{\downarrow}+j}))$, but the specific value is not important.} The highest cost and points (for $A^{\downarrow}$) are provided by the order wherein members of $A^{\downarrow}$ are placed directly below $c_1$.  The cost in this case is some $\overline{B}^{\downarrow}_{O_i}(\ell^{\downarrow})$.  Importantly, every increase in the cost is also an increase in Borda points. So, the total number of points awarded by this order is $\underline{M}^{\downarrow}_{O_i}(\ell^{\downarrow})+\overline{B}^{\downarrow}_{O_i}(\ell^{\downarrow})$.  
By shifting the elements of $A^{\downarrow}$ up, one step at a time, without ever swapping elements of $A^{\downarrow}$, 
any intermediate cost and points values can be obtained.  That is, for any $t^{\downarrow}=0,\ldots, \overline{B}^{\downarrow}_{O_i}(\ell^{\downarrow})$, there is an order of $A^{\downarrow}$ with all elements are below $c_1$, the is cost $t^{\downarrow}$, and the Borda points are $\underline{M}^{\downarrow}_{O_i}(\ell^{\downarrow})+t^{\downarrow}$.

\paragraph{$c_1$:} The cost due to placing $c_1$ at $m-k_1$ is $B^{1}_{O_i}(k_1)=pos(O_i,c_1)-(m-k_1)$, and the points provided are $k_1$.

\paragraph{$A^{\uparrow}$:}  
The least cost and least points (for $A^{\uparrow}$) are afforded by the order wherein members of $A^{\downarrow}$ are placed as low as possible, provided that: (i) all are above $c_1$, (ii) the order among $A^{\uparrow}$ is as in $O_i$, (iii) no element is placed lower than in $O_i$. Let $\underline{B}^{\uparrow}_{O_i}(\ell^{\uparrow})$ and $\underline{M}^{\uparrow}_{O_i}(\ell^{\uparrow})$ the cost and the Borda points of this order. The highest cost and points (for $A^{\uparrow}$) are provided by the order wherein members of $A^{\uparrow}$ are placed at the top locations. Denote the cost and points of this order by $\overline{B}^{\uparrow}_{O_i}(\ell^{\uparrow})$ and $\overline{M}^{\uparrow}_{O_i}(\ell^{\uparrow})$, respectively.  Again, by shifting the elements of $A^{\uparrow}$ one step at a time, without ever swapping elements of $A^{\uparrow}$, 
any intermediate cost and point values can be obtained.  That is, for any $t^{\downarrow}=0,\ldots, (\overline{B}^{\uparrow}_{O_i}(\ell^{\uparrow})-\underline{B}^{\uparrow}_{O_i}(\ell^{\uparrow}))$, there is an order of $A^{\downarrow}$, with all elements above $c_1$, with cost $\underline{B}^{\uparrow}_{O_i}(\ell^{\uparrow})+t^{\downarrow}$ and points $\underline{M}^{\uparrow}_{O_i}(\ell^{\uparrow})+t^{\downarrow}$.

So, in all, we have that $k_1,k_{A_{-1}}$ is realizable with least cost $B'$ if there exists $0\leq \ell^{\downarrow},\ell^{\uparrow}\leq h,0\leq t^{\downarrow}\leq \overline{B}^{\downarrow}_{O_i}(\ell^{\downarrow}),0\leq t^{\uparrow}\leq (\overline{B}^{\uparrow}_{O_i}(\ell^{\uparrow})-\underline{B}^{\uparrow}_{O_i}(\ell^{\uparrow})),\underline{M}^{\downarrow}_{O_i}(\ell^{\downarrow})\leq k^{\downarrow}_{A}\leq \underline{M}^{\downarrow}_{O_i}(\ell^{\downarrow})+\overline{B}^{\downarrow}_{O_i}(\ell^{\downarrow}),\underline{M}^{\uparrow}_{O_i}(\ell^{\uparrow})\leq k^{\uparrow}_A\leq \overline{M}^{\uparrow}_{O_i}(\ell^{\uparrow})$ such that:
\begin{enumerate}
    \item $\ell^{\downarrow}+\ell^{\uparrow}=|A_{{-1}}|$,
    \item $\ell^{\uparrow}<m-k_1-1$ (there is sufficient room above $c_1$),
    \item $k_A=k_A^{\downarrow}+k_A^{\uparrow}$,
    \item $k_A^{\downarrow}=\underline{M}^{\downarrow}_{O_i}(\ell^{\downarrow})+t^{\downarrow} \;\; , \;\;\; k_A^{\uparrow}=\underline{M}^{\uparrow}_{O_i}(\ell^{\uparrow})+t^{\uparrow}$,
    \item $B'=t^{\downarrow}+B^{1}_{O_i}(k_1)+\underline{B}^{\uparrow}_{O_i}(\ell^{\uparrow})$.
\end{enumerate}
Then, $h(i,k_{A_{-1}},k_1)$ is the least $B'$ for which the above hold. There are a polynomial number of combinations to check, each of which is checked in constant time. The values of $\overline{B}^{\downarrow}_{O_i},\underline{B}^{\uparrow}_{O_i},\overline{B}^{\uparrow}_{O_i},\underline{M}^{\downarrow}_{O_i}, \underline{M}^{\uparrow}_{O_i},\overline{M}^{\uparrow}_{O_i}$ can be computed in polynomial time. So, the entire computation is polynomial.
\end{proof}
\begin{theorem}
    Borda$_{0}$-CBP can be solved in polynomial time for $1$-bribery, $\$$-bribery and coalition-shift-bribery. 
  \label{thm:dolar-Borda-CBPP}
\end{theorem}
\begin{proof}
Let $f(i,k_{A},k_1)$ be the least cost of a bribe $\hat{\order}_{\{1,\ldots,i\}}=(\hat{O}_1,\ldots,\hat{O}_i)$ of voters $v_1,\ldots, v_i$, with ${Borda}(\hat{\order}_{\{1,\ldots,i\}},A)=k_A$ and  $Borda(\hat{\order}_{1,\ldots,i},c_1)=k_1$.  Then, $f$ can be computed as follows:
    \begin{align*}
    \bullet &    f(1,k_A,k_1)=h(1,k_A-k_1,k_1)\\
    \bullet &    f(i,k_A,k_1)= 
        \min_{0\leq k'_A\leq k_A, 0\leq k'_1\leq k_1} \\ & \{ h(i+1,k'_A-k'_1,k'_1)+f(i,k_A-k'_A,k_1-k'_1) \}
    \end{align*}
Now, $f(n,k_{A},k_1)$ is the cost of the least bribe over all voters that obtains $k_A$ points for $A$ and $k_1$ points for $c_1$.
So, in order to solve Borda$_{0}$-CBP, 
compute $f(n,k_A,k_1)$ for all $1\leq k_1\leq k_A \leq \frac{|A|m}{2}$, such that $\frac{k_A}{n{{m(m-1)}\over 2}}\geq \varphi$ (the denominator is the total number of Borda points), and $\frac{k_1}{k_A}\geq \rho$. If any has value $\leq B$ return success, otherwise return failure.  There is only a polynomial number of such pairs $k_1,k_A$ to check.
\end{proof}
\begin{corollary}
 Borda$_{0}$-CB can be solved in polynomial time for $1$-bribery, $\$$-bribery and coalition-shift-bribery
\label{cor:1-dolar-shift-Borda-CBP}
\end{corollary}

\paragraph{Swap Bribery}
For swap-bribery,  Borda$_{0}$-CB is NP-hard.  The proof is by a reduction from Min-Bisection problem. Hence, Borda$_{0}$-CBP is also NP-hard.
\begin{theorem}
     Borda$_{0}$-CB and Borda$_{0}$-CBP are NP-hard for swap-bribery. 
    \label{thm:swap-borda-cbp}
\end{theorem}

\subsection{Borda$_t$-CB}
Borda$_t$-CB is NP-hard for all bribery types studied in this paper.
For $1$-bribery and $\$$-bribery we show a reduction from the $3$-$4$-Exact-Cover problem, inspired by the reductions of~\cite{elkind2009swap,put2016complexity}.
\begin{theorem}
    Borda$_t$-CB is NP-hard for $1$-bribery and $\$$-bribery.
    \label{thm:1-Borda-t-cbp}
\end{theorem}
Theorem $5$ of ~\cite{put2016complexity} establishes the NP-hardness for shift-bribery under Borda with an election threshold and a single preferred candidate.  This implies NP-hardness for our setting for both coalition-shift-bribery and swap-bribery.
\begin{corollary}
     Borda$_t$-CB is NP-hard for coalition-shift-bribery and swap-bribery.
    \label{cor:shift_swap_non-Borda-t-cbp}
\end{corollary}
Since Borda$_t$-CB is a special case of Borda$_t$-CBP: 
\begin{corollary}
     Borda$_t$-CB is NP-hard for $1$-bribery, $\$$-bribery, coalition-shift-bribery, and swap-bribery.
    \label{cor:all-Borda-t-ccbp}
\end{corollary}

\section{Conclusions and Future Work}
 We introduced and studied the complexity of bribery in parliamentary elections when the goal is not (only) to promote a specific party, but rather an entire coalition of parties.  
We introduced two key problems, the Coalition-Bribery problem (CB) and the Coalition-Bribery-with-Preferred-party-Problem (CBP), both extending the classic bribery by incorporating coalition objectives. Our study addressed both Plurality and Borda scoring rules, as well as the impact of threshold constraints in the seat allocation rule.

 An obvious future direction is to extend the study to other multi-winner voting rules (e.g.~ STV, PAV, Cumulative Voting, Limited Voting, ...). Additionally, we considered a single briber (as is customary).  The case with multiple bribers interacting in a strategic game is an interesting and realistic question.  In this case, an equilibrium analysis is called for. 
 
More broadly, in this paper we only considered bribery. It would be interesting to study related notions, such as \emph{control}, \emph{manipulation}, and \emph{gerrymandering},  in the context of coalitions.  
\section{Acknowledgments}
This research is partly supported by the Israel Science Foundation grants 2544/24, 3007/24 and 2697/22.

\bibliographystyle{acm}
\bibliography{ijcai25}

\clearpage
\appendix
\section{Table of Notations}
\begin{table}[h!]
    \centering
    \renewcommand{\arraystretch}{1.2} 
    \begin{tabular}{|p{3cm}|p{5cm}|} 
        \hline
        \hline
        \textbf{Notation} & \textbf{Description} \\ \hline
        $E = (C, V, \mathcal{O})$ & A parliamentary election, where $C$ is the set of political parties, $V$ is the set of voters, \\
                                  & and $\mathcal{O}$ is the set of voters' preference orders. \\ \hline
        $C = \{c_1, \ldots, c_m\}$ & Set of political parties (candidates). \\ \hline
        $V = \{v_1, \ldots, v_n\}$ & Set of voters. \\ \hline
        $\mathcal{O} = \{O_1, \ldots, O_n\}$ & Sequence of voters' preference orders. \\ \hline
        $pos_{O_i}(c_j)$ & Position of party $c_j$ in voter $v_i$'s preference order $\mathcal{O}_i$. \\ \hline
        $\gamma$ & Positional scoring function. \\ \hline
        $\gamma(\mathcal{O}, c_j)$ & Total points attained by party $c_j$ under orders $\mathcal{O}$. \\ \hline
        $\seats{\gamma}{t}$ & Fraction of seats allotted to a party under scoring function $\gamma$ and threshold $t$. \\ \hline
        $t$ & Threshold fraction below which parties are inactive. \\ \hline
        $\bbribe$ & Sequence of modified preference orders after the bribe. \\ \hline
        $\pi_i(\mathcal{O}_i')$ & Cost of modifying voter $v_i$'s preference order from $\mathcal{O}_i$ to $\mathcal{O}_i'$. \\ \hline
        $\pi(\bbribe)$ & Total cost of a bribe $\bbribe$. \\ \hline
        $A = \{c_1,\ldots,c_h\}$ & Coalition of parties in the CB and CBP problems. \\ \hline
        $\varphi$ & Target support fraction for the coalition in CB and CBP. \\ \hline
        $\rho$ & Target ratio for a preferred party's share within the coalition in CBP. \\ \hline
        $A_{-1}$ & $ A\setminus\{c_1\}$\\ \hline
        $\bar{A}$  & $C\setminus A$ \\ \hline
        $top(O_i)$ &  the first party in $O_i$. \\ \hline
        $B$ & The budget \\ \hline
    \end{tabular}
    \caption{Summary of Notations}
    \label{tab:notations}
\end{table}

\section{Example}
Here we provide an additional example to emphasize the differences between our coalition bribery with a preferred party problem and the bribery problem in a parliamentary election for a single preferred candidate, as studied in the literature~\cite{put2016complexity}.

We start with a simple example. 
\begin{Example}
    The input to the CB is the election $E=(C,V,O)$, 
    Where $C = \{c_1, c_2, c_3, c_4\}$, $V = \{v_1,v_2,v_3,v_4\}$ and the preference orders are $O_i= c_4\succ c_3 \succ c_2 \succ c_1$. 
    The coalition $A =\{c_1,c_2\}$.
    The bribery type is $1$-bribery, and the budget is $B = 1$.
    The target support $\varphi = \frac{1}{4}$. 

    Assume that the seat allocation function is $\seats{\text{Plurality}}{0}$, $\seats{\text{Plurality}}{0}(\mathcal{O}, A) = 0$.
    Then bribe $v_1$ to change his preference order to have $c_1$ or $c_2$ in the first place, is sufficient to have that $\seats{\text{Plurality}}{0}(\bbribe, A) = \frac{1}{4} \geq \varphi$.

     Assume that the seat allocation function is $\seats{\text{Borda}}{0}$, $\seats{\text{Borda}}{0}(\mathcal{O}, A) = \frac{4}{24}$.
    Then bribe $v_1$ to change his preference order to have $c_1$ and $c_2$ in the first and second places, is sufficient to have that $\seats{\text{Borda}}{0}(\bbribe, A) = \frac{8}{24} \geq \varphi$.
\end{Example}

\begin{Example}
Assume that the seat allocation function is $\seats{\text{Plurality}}{t}$, and the election threshold is $\frac{1}{8}$.
Let the budget be $B=3$. 

In the case of a single preferred party (candidate), the desired fraction is $\frac{3}{8}$,  the preferred party is $c_1$, and the bribery type is shift-bribery. 

In the CBP let the target support $\varphi=\frac{1}{2}$, and the target ratio  $\rho=\frac{3}{4}$. The bribery type is coalition-shift-bribery. The coalition is $A=\{c_1,c_2\}$, and the preferred party is $c_1$.

Let the set of parties be $C = \{c_1, c_2, c_3\}$ and the set of voters be $V=\{v_1, v_2, \ldots v_{16}\}$.

The preference orders $O_1, O_2,\ldots, O_5$ is: $c_1 \succ c_2 \succ c_3$.
The preference order $O_6$ is: $c_2 \succ c_3 \succ c_1$.
The preference orders $O_{7}, O_{8}, \ldots, O_{16}$ is: $c_3 \succ c_1 \succ c_2$.

The price function for $v_1, v_7, v_8$ is $1$ for a single shift and $5$ for two shifts.
The price function for all other voters is $5$ per shift. 

In the parliamentary election with a single preferred party, for successful bribe, one of the voters who originally voted for $c_3$ must be bribed to vote for $c_1$ (i.e., either $v_7$ or $v_8$, as the budget is insufficient to bribe other voters).

However in CBP, two voters who originally voted for $c_3$ must be bribed to vote for $c_1$ (i.e., $v_{7}$ and $v_{8}$), resulting  in a coalition of size $\frac{7}{16}$ which is less than $\varphi$.
Then, $v_1$ who originally voted for $c_1$ needs to be bribed to vote for $c_2$.
This allows $c_2$ to pass the election threshold, enabling the coalition to achieve $\frac{6+2}{16} = \varphi$, while $c_1$ gets $\frac{3}{8}$ from the coalition. 
\end{Example}

\section{Proof of Theorem \ref{thm:shift-Plurality-t-CBP}}
\begin{definition}[3-4-Exact-Cover problem]
    Given a set $Z = \{z_1,\ldots z_n\}$ with $n$ elements.
    A collection $D$ of $4$-elements subsets of $Z$.
    That is,  $D=\{D_1, D_2, \ldots D_m\}$ such that each $D_i\subseteq Z$ and $|D_i|=4$.
    Moreover, each $z_i \in Z$ belongs to exactly $3$ subsets in $D$. 
    The goal is to find sub-collection $D'\subseteq D$ such that every element in $Z$ is included in exactly one subset in $D'$.
\end{definition}
\begin{reminder}{\ref{thm:shift-Plurality-t-CBP}}
       Plurality$_t$-CB with coalition-shift-bribery is NP-hard, even when the price function is such that all $s_i(x)$'s are linear.
\end{reminder}
\begin{proof}
    Given an instance of the 3-4-Exact-Cover problem, where $Z=\{z_1,z_2,\ldots z_n\}$ is a set of elements and $D=\{D_1, D_2, \ldots D_m\}$  is a collection of $4$-element subsets of $Z$.
    Construct an instance of Plurality$_t$-CB such that there is a successful bribe if and only if an exact cover exists. 
    
    \textbf{The election.}
    \sloppy{The election $E=(C,V,O)$ is defined as follows:
    The set of parties is $C=\{c_1, c_2,\ldots c_m, c_{m+1}, c_{m+2}, \ldots, c_{m+n}\ldots c_{m+3n +1}\}$.}
    Each party $c_j$ for $1\leq j\leq m$ corresponds to the subset $D_j\in D$. 
    The set of voters is $V = \{v_1,v_2,\ldots, v_n, v_{n+1}, \dots v_{2n}\}$.
    Each voter $v_i$ for $1\leq i\leq n$, corresponds to the element $z_i \in Z$.
    The election threshold be $t=\frac{4}{2n}$.
    
    \textbf{Voter preference orders.}
    \begin{itemize}
        \item For each  $1\leq i\leq n$, the preference order $O_i$ is: party $c_{m+1}$ in the first place, followed by the $3$ parties corresponding to the subsets that contain $z_i$, then the parties $c_j$ for $m+2 \leq j \leq m+3n +1$, and then all other parties in an arbitrary order.  
        \item For each $n+1\leq i\leq 2n$, the preference order $O_i$ is: party $C_{m+1}$ in the first, followed by the parties $c_j$ for $m+2 \leq j \leq m+3n+1$, and then all other parties in an arbitrary order. 
    \end{itemize}
    
    \textbf{CB.}
    Let the coalition $A =\{c_1, c_2,\ldots c_m\}$.  The bribery type is coalition-shift-bribery and the price function for each voter is $s_i(x) = x$.
    The Budget $B=3n$, and the target support is $\varphi = \frac{1}{2}$.

    \textbf{Constraints.}
    \begin{itemize}  
        \item  For each voter $v_i$ where $n+1\leq i\leq 2n$, there are $3n+1$ parties before the parties from $A$, so it is impossible to bribe these voters.
        Therefore, there are at least $n$ votes for $c_{m+1}$.
        \item  For each voter $v_i$ where $1\leq i\leq n$, it is possible only to bribe them in such a way that $v_i$ votes for one of the parties corresponding to a subset $D_j$ that contains $z_i$, since there are $3n+4$ parties before the other parties from $A$.
        \item In order to have $\sum_{a \in A} (R(bribe(E,seq), a)) \geq \varphi = \frac{1}{2}$, there must be $n$ votes for $a \in A$.
         \item  There are $2n$ voters. Since $t=\frac{4}{2n}$, only parties that receive at least $4$ votes pass the election threshold.
         \item A party $c_i$ will be part of the parliament only if all four voters corresponding to the elements of the set $D_i$ change their preference order to rank  $c_i$ in the first place.
         \item  It is possible to bribe all voters $v_i$, for $1\leq i\leq n$, to vote for a party corresponding to a subset $D_j$ that contains $z_i$, as the cost of each such bribe is at most $3$.
    \end{itemize}
    
    \textbf{To sum up.}
    Assume that there is a successful bribe. 
    Since $|D_j|= 4$, each $c_j$ can receive at most $4$ votes after the bribe. 
    In order to satisfy $\sum_{a \in A} (R(bribe(E,seq), a)) \geq\frac{1}{2}$, there must be $\frac{n}{4}$ parties from $A$ that receive $4$ votes.  These  $\frac{n}{4}$ parties correspond to $\frac{n}{4}$ subsets that form an exact cover.

    On the other hand, assume that there is an exact cover. A successful bribe would involve bribing each voter $v_i$ for $1\leq i\leq n$, to vote for the party corresponding to the subset that covers $z_i$. 
\end{proof}

\section{proof of theorem \ref{thm:1-Borda-t-cbp}}
For $1$-bribery and $\$$-bribery, Borda$_t$-CB is an NP-hard problem we show a reduction from the $3$-$4$-Exact-Cover problem. This reduction is inspired by the reductions from~\cite{elkind2009swap,put2016complexity}.
\begin{reminder}{\ref{thm:1-Borda-t-cbp}} 
    Borda$_t$-CB is an NP-hard problem for $1$-bribery and $\$$-bribery.
\end{reminder}
$S$
\begin{proof}
    Given an instance of the $3$-$4$-Exact-Cover problem. 
    Let $Z=\{z_1,z_2,\ldots,z_n\}$ be the set of elements, and let  $D=\{D_1,D_2,\ldots,D_m\}$ be the collection of subsets.
    Construct an instance of Borda$_t$-CB such that there is a successful bribe if and only if an exact cover exists.

    \textbf{The election.}
    \sloppy{Let the set of parties be: $C=\{c_1,c_2,\ldots,c_{mn+1}, u_1, u_2,\ldots, u_n\}$. Let $U = \{u_1, u_2,\ldots, u_n\}\subset C$ be a subset of parties, where each party $u_j$ is corresponding to element $z_j\in Z$.}
    
    Let the preferred coalition of parties be $A=\{c_1,c_2,\ldots,c_{mn+1}\}$. 
    Let the set of voters be $V=\{v_1,v_2,\ldots, v_m, v'_1,v'_2,\ldots, v'_m\}$,
    where voters $v_i, v'_i$ correspond to the subset $D_i$.
    
    Let $F_i$ be the parties from $U$ corresponding to the elements in $D_i$, arranged in a fixed order, and let $\overline{F_i}$ represent the same elements in reverse order. 
    Similarly, $U\setminus F_i$ and $\overline{U\setminus F_i}$ represent the parties corresponding to elements of $Z$ not in $D_i$ in a fixed order and in reverse order, respectively. 
    We slightly abuse notation and use $A$ to refer both to the set of coalition parties and to the same parties in ascending order and $\overline{A}$ be these parties in a reverse order.
    
    \textbf{Voter preferences.}
    The preference order of $v_i$ is: $F_i \succ A \succ U\setminus F_i$.
    The preference order of $v'_i$ is: $\overline{A} \succ \overline{U\setminus F_i}\succ \overline{F_i}$.
    
    Consider party $c_1$.
    From each voter $v_i$, $c_1$ receives $mn+1+n-5$ points, since $c_1$ ranked after the parties from $F_i$.
    From each voter $v'_i$, $c_1$ receives $n$ points, since only $n$ parties ranked after him. 
    Thus, $c_1$ receives $m\cdot (mn+n-4)+m\cdot n$.
   
    Since the order of parties from $A$ in $v_i$'s preference order is the reverse of their order in  $v'_i$'s preference order, and both orders rank all parties from $A$ consecutively, the number of points each party from $A$ receives is identical.
    Therefore, each $c_i\in A$ receives $m^2n+2mn-4m$ points.

    Consider a party $u_j \in U$.
    For the same reasons as above, the position of $u_j$ within the chosen order does not affect the total number of points it receives from $v_i$ and $v'_i$ combined. What matters is only whether $z_j\in D_i$ (whether the element corresponding to $u_j$ is in the subset associated with these voters).  
    There are three subsets $D_i$  such that $z_j\in D_i$. 
    From each pair of voters, $v_i, v'_i$, corresponding to these subsets, the party $u_j$ receives $mn+n+1$ points. Since, there are three such subsets, $u_j$ receives $3(mn+n+1)$ points from these six voters. 
    From each other pair of voters $v_i,v'_i$, where $z_j\not\in D_i$, the party $u_j$ receives $n$ points.
    Since there are $n-3$ such subsets, $u_j$ receives $n^2-3n$ points from these $2n-6$ voters.  
    
    Thus, we have that: 
    \begin{enumerate}
        \item Each $c_i\in A$ gets: $m^2n+2mn-4m$ 
        \item Each $u_i \in U$ gets: $n^2+3mn+3$
    \end{enumerate}

    \textbf{CB.}
    Let $T$ be the minimum number of points a party must receive to pass the election threshold.  
    Let $\varphi = 1$. To ensure that the coalition $A$  gets $100\%$ of seats in the parliament, the bribe must reduce the points of each party  $u_i\in U$ to less than $T$, while ensuring that at least one party from $A$ receives more than $T$. 

    Set $T$ to be $n^2+2mn+2$, meaning that the bribe needs to reduce the score of each $u_i$ by $mn + 1$ points. 

    In the $3$-$4$-exact-cover problem, it holds that $m=\frac{3n}{4}$.
    Therefore, for each $n\geq 4$, the parties from the coalition $A$ will pass the election threshold.
    The budget is set to $B=\frac{n}{4}$, and the price function of all voters is $\pi_i=1$.

    \textbf{To sum up.}
    We will show that if there exists an exact cover, then a successful bribe is possible. Conversely, if there is no exact cover, then no successful bribe can be achieved.

    First, assume that there is an exact cover. 
    Let $D'\subset D$, be the exact cover.
    It holds that $|D'| = \frac{n}{4}$, and we denote $D' = \{d_{i_1}, d_{i_2}, \ldots ,d_{i_{\frac{n}{4}}}\}$.
    The bribe that bribes the corresponding voters $v_{i_1}, v_{i_2}, \ldots ,v_{i_{\frac{n}{4}}}$ by changing the preference order from $F_i \succ A \succ U\setminus F_i$ to $A \succ F_i \succ U\setminus F_i$ is a successful bribe. Since each element appears in exactly one of the subsets in the exact cover, the corresponding party will be in exactly one of the bribed voters' preference orders in a way that changed his position, causing them to lose exactly $mn+1$ points. This results in those parties receiving one point less than $T$.

    On the other hand, assume that there is no exact cover, and by contradiction, suppose there is a successful bribe.
    Hence, at most $\frac{n}{4}$ voters bribed in the bribe, which implies that there must be at least one $z_i$ such that the voters corresponding to the subset containing $z_i$ are not bribed. 
    Thus, each change in a voter’s preference order can decrease his score by at most $n$ points. The total decrease is therefore bounded by $\frac{n^2}{4}$. 
    However, since $\frac{n^2}{4} < mn+1 = \frac{3\cdot n^2}{4} + 1$, this party would still pass the threshold. 
    And therefore, this is a contradiction to the assumption that it is a successful bribe.
\end{proof}

\section{proof of theorem \ref{thm:swap-borda-cbp}}
For swap-bribery, Borda-CB is an NP-hard problem. 
We prove it by a reduction from the NP-hard Min-Bisection problem \cite{garey1979computers}.

\begin{definition}[Min-Bisection problem]
    Given a graph $G = (V,E)$, such that $|V| = 2n$ and a number $k\in \mathbb{N}$.
    The goal is to determine whether it is possible to partition the vertices of the graph into two disjoint sets $V_1$ and $V_2$, such that $|V_1| = |V_2| = n$ and there are at most $k$ edges crossing between $V_1$ and $V_2$.
\end{definition}
\begin{reminder}{\ref{thm:swap-borda-cbp}}
    Borda$_{0}$-CB is an NP-hard problem for swap-bribery. 
\end{reminder}
\begin{proof}
    Given an instance of the Min-Bisection problem with a graph $G = (U, E)$ and a bound $k$.
    Construct an instance of Borda-CB such that a successful bribe can be achieved if and only if there is a feasible solution to the Min-Bisection problem for $G$ and $k$.

    \sloppy{The set of parties contains two parties, $c_{1,i}$ and $c_{2,i}$, for each $u_i \in U$ and additional one vertex, $x$. That is $C = \{c_{1,1}, c_{1,2} \ldots, c_{1,2n}, c_{2,1}, c_{2,2} \ldots, c_{2,2n}, x\}$, the preferred coalition of parties is $A = \{c_{1,1}, c_{1,2} \ldots, c_{1,2n}, c_{2,1}, c_{2,2} \ldots, c_{2,2n} \}$.
    Let $C^1 = \{c_{1,1}, c_{1,2} \ldots, c_{1,2n}\}$ and $C^2 = \{c_{2,1}, c_{2,2} \ldots, c_{2,2n}\}$.
    The set of the voters is $V= {v_1}$.}
    
    \textbf{Voter's preference order.}
    The preference order of $v_1$ is: $x \succ c_{1,1} \succ c_{1,2} \ldots \succ c_{1,2n} \succ  c_{2,1}\succ c_{2,2} \ldots \succ c_{2,2n}$
    
    \textbf{Voter's price function.}
     the price function of $v_1$ is as follows:
     \begin{itemize}
        \item $pi_1^{\{c_{1,i},c_{1,j}\}} = 0$
        \item $pi_1^{\{c_{2,i},c_{2,j}\}} = 0$
        \item $\pi_1^{\{c_{2,i}, c_{1,i}\}} = B+1$
        \item for $i\not=j$, $pi_1^{\{c_{2,i}, c_{1,j}\}} = 1$ if $(i,j)\in E$ and $0$ otherwise
        \item  $pi_1^{\{c_{1,i}, x\}} = (kn)^2$
        \item  $pi_1^{\{c_{2,i}, x\}} = (kn)$
     \end{itemize}
     
    \textbf{CB.}
    The budget is $B = n\cdot(k+n)^2+n\cdot k\cdot n + k$.
    Let the target support, $\varphi$, be $\frac{4n}{4n+1}$. That is, a successful bribe achieves at least $2n$ new points for the coalition within the budget.

    \textbf{To sum up.}
    We will show that there is a successful bribe if and only if a feasible solution exists to the Min-Bisection problem.

    First, assume that a feasible solution exists to the Min-Bisection problem.
    Then let $U_1$, $U_2$, be the partition of $U$ in a feasible solution.
    
    For each $u_i \in U_1$ first shift all corresponding parties from $C^1$ to a position before  $x$, then shift all corresponding parties from $C^2$ to a position before  $x$.
    As a result, in the new preference order, the $c_{1, i}$  and $c_{2, i}$ parties corresponding to $U_1$ will appear sequentially before $x$, while those corresponding to $U_2$ will remain after $x$. Thus, the new preference order will be:
    $c_{1,i_1}\succ \ldots \succ c_{1,i_n}\succ c_{2,i_1} \succ \ldots \succ c_{2,i_n}  \succ  x \succ c_{1,j_1}\succ \ldots \succ c_{1,j_n} c_{2,j_1}\succ \ldots \succ c_{2,j_n}$.
    
    The size of $U_1$ is $n$, and each swap between a party and $x$ achieves one new point for the coalition. 
    Hence, in this bribe, the coalition achieves $2n$ points by the bribe. 
    
    The cost of such a bribe consists of the following:
    \begin{enumerate}
        \item Shift of $c_{1,i}$ for each $u_i \in U_1$: each such shift costs $(k+n)^2$, since only the swap with $x$ cost $(kn)^2$ and all other saps are free. There are $n$ parties of this type that are shifted, hence the total cost for all $n\cdot(k+n)^2$.
        \item Shifts of $c_{2,i}$ for each $u_i \in U_1$: each such shift cost $k\cdot n$ for the swap with $x$ plus $1$ for each edge that $u_i$ has to vertices in $U_2$, i.e., it add to the cost: $|\{j\in U_2| (i,j)\in E\}|$. There are $n$ parties of this type and since the partition is according to a feasible solution of the Min-Bisection problem the total number of such edges is at most $k$, yielding a total cost of $n\cdot k\cdot n + k$.
    \end{enumerate}

    Thus, the combined cost is at most $n\cdot(k+n)^2+n\cdot k\cdot n + k$ and does not exceed $B$.

    On the other hand, assume that there is no feasible solution to the Min-Bisection problem.
    The party $x$ is the only one not included in the coalition $A$, meaning that only a swap between $x$ and some other party contributes a point to the coalition. 
    The budget $B$ allows a party from $c^2$ to be shifted before $x$ only after the corresponding party from $c^1$ has been shifted before $x$. 
    In addition, the budget only allows shifting up to $n$ parties from $c^1$.
    Hence in order to achieve $n$ new points, the bribe must shift $n$ parties from $C^1$, and then shift the $n$ corresponding parties from $C^2$.
    
    For each subset, $C'\subset C^1$, of $n$ parties, let $U_1$ and $U_2$ be the corresponding partitions of the vertices of $U$.
    Shifting the parties from $C'$ allows the shifting of their corresponding parties from $C^2$; however, this incurs additional costs, specifically, an additional cost of $1$ for each edge between $U_1$ and $U_2$.  
    Since a feasible solution to the Min-Bisection problem does not exist,  every possible subset $C'$ yields more than $k$ edges between $U_1$ and $U_2$.
    Thus, the excessive edge costs increase the total cost beyond $B$, making it impossible to shift all corresponding parties from $C^2$.
    As a result, there is no successful bribe.   
    
\end{proof}
 Borda$_{0}$-CBP is an extension of Borda$_{0}$-CB, hence Borda$_{0}$-CBP is also an NP-hard problem.

\end{document}